\documentclass[letterpaper, 10pt, conference]{ieeeconf}  

\usepackage{graphicx}
\usepackage{subfig}
\usepackage{amsmath}
\usepackage{amsthm} 
\usepackage{amssymb}
\usepackage[font=footnotesize]{caption} 
\usepackage{subfig} 
\usepackage[noadjust]{cite} 
\usepackage{color}
\usepackage{algorithm} 
\usepackage{algpseudocode} 
\usepackage{enumerate}
\usepackage[hidelinks,colorlinks=false]{hyperref}
\usepackage{setspace}

\usepackage{tikz}
\usetikzlibrary{calc} 
\usetikzlibrary{shapes} 
\usetikzlibrary{chains}
\usetikzlibrary{fit}
\usetikzlibrary{arrows}
\usetikzlibrary{decorations.text} 
\usetikzlibrary{decorations.markings}

\newtheorem{theorem}{Theorem}
\newtheorem{lemma}{Lemma}
\newtheorem{assumption}{Assumption}

\newtheorem{definition}{Definition}
\newtheorem{problem}{Problem}

\newcommand{\Null}{\mathrm{Null}}
\newcommand{\Range}{\mathrm{Range}}
\newcommand{\one}{\mathbf{1}}
\newcommand{\rank}{\mathrm{rank}}
\newcommand{\myspan}{\mathrm{span}}
\newcommand{\mydiag}{\mathrm{diag}}

\newcommand{\blkdiag}{\mathrm{blkdiag}}

\newcommand{\T}{\mathrm{T}}

\newcommand{\R}{\mathbb{R}}
\newcommand{\G}{\mathcal{G}}
\newcommand{\E}{\mathcal{E}}
\newcommand{\V}{\mathcal{V}}
\newcommand{\N}{\mathcal{N}}

\renewcommand{\L}{\mathcal{L}}

\newcommand{\dia}[1]{\mathrm{diag}\left(#1\right)} 

\graphicspath{{figures/}}

\newcommand{\leftm}{\left[\begin{array}}
\newcommand{\rightm}{\end{array}\right]}

\begin{document}

\title{Bearing-Based Distributed Control and Estimation of Multi-Agent Systems}
\author{Shiyu Zhao and Daniel Zelazo
\thanks{S. Zhao and D. Zelazo are with the Faculty of Aerospace Engineering, Technion - Israel Institute of Technology, Haifa, Israel.
    {\tt\small szhao@tx.technion.ac.il, dzelazo@technion.ac.il}}
}
\IEEEoverridecommandlockouts
\maketitle
\begin{abstract}
This paper studies the distributed control and estimation of multi-agent systems based on \emph{bearing} information.
In particular, we consider two problems: (i) the distributed control of bearing-constrained formations using relative position measurements and (ii) the distributed localization of sensor networks using bearing measurements.
Both of the two problems are considered in arbitrary dimensional spaces.
The analyses of the two problems rely on the recently developed bearing rigidity theory.
We show that the two problems have the same mathematical formulation and can be solved by identical protocols.
The proposed controller and estimator can globally solve the two problems without ambiguity.
The results are supported with illustrative simulations.
\end{abstract}
\overrideIEEEmargins

\section{Introduction}

In recent years, distance rigidity theory has played an important role in the area of distributed control and estimation of multi-agent systems. For example, it has been widely applied in distance-based formation control \cite{Krick2009IJC,Dorfler2009ECC,oh2013IJRNC,SunZhiyong2014IFAC} and distance-based network localization \cite{Anderson2006NetworkLocalization,Mao2007NetworkLocalization}.
The \emph{global distance rigidity} can ensure the unique shape of a framework, but it is very difficult to examine mathematically.
Many of the existing works adopted the assumption on \emph{infinitesimal distance rigidity} which can be easily examined by a rank condition \cite{Krick2009IJC,Dorfler2009ECC,oh2013IJRNC,SunZhiyong2014IFAC}.
Infinitesimal distance rigidity is, however, not able to ensure a unique shape, which may result in undesired control or estimation solutions.
Additionally, due to a large number of undesired equilibriums in distance-based formation control, the gradient control law can only be proved to be locally stable for general formations \cite{Krick2009IJC,Dorfler2009ECC,oh2013IJRNC,SunZhiyong2014IFAC}.

As with distances, bearings can also be used to characterize the shape of a network.
In recent years, there has been a growing interest in bearing rigidity theory (also known as parallel rigidity theory) \cite{eren2003,zelazo2014SE2Rigidity,bishopconf2011rigid,Eren2012IJC} and bearing-based control and estimation problems including bearing-based formation control \cite{bishop2010SCL,Eren2012IJC,bishopconf2011rigid,zhao2013SCLDistribued,Eric2014ACC} and bearing-based network localization \cite{Piovan2013Automatica,ZhuGuangwei2014Automatica}.
Most of the previous studies on bearing rigidity only focused on frameworks in two-dimensional ambient spaces \cite{eren2003,bishopconf2011rigid,Eren2012IJC,zelazo2014SE2Rigidity}.
In our recent work \cite{zhao2014TACBearing}, we extended the previous studies and established the theory of bearing rigidity in arbitrary dimensions.
We showed that bearing rigidity has a number of attractive features compared to distance rigidity.
For example, infinitesimal bearing rigidity can globally determine the unique shape of a framework, and it can also be conveniently examined by a rank condition of the bearing rigidity matrix.
The bearing rigidity theory has been successfully applied to solve bearing-only formation control problems in \cite{zhao2014TACBearing}.

In this paper, we apply bearing rigidity theory to solve two problems: (i) distributed control of bearing-constrained formations using relative position measurements and (ii) distributed localization of sensor networks using bearing measurements. Both of the problems are considered in arbitrary dimensions.
The problem formulations of the two problems actually are the same and hence the two problems can be solved by identical protocols.
The proposed linear controllers and estimators can globally solve the two problems without ambiguity.


\paragraph*{Notations}

Given $A_i\in\mathbb{R}^{p\times q}$ for $i=1,\dots,n$, denote $ \mydiag(A_i)\triangleq\blkdiag\{A_1,\dots,A_n\}\in\mathbb{R}^{np\times nq}$.
Let $\Null(\cdot)$ and $\Range(\cdot)$ be the null space and range space of a matrix, respectively.
Denote $I_d\in\R^{d\times d}$ as the identity matrix, and $\one\triangleq[1,\dots,1]^\T$.
Let $\|\cdot\|$ be the Euclidian norm of a vector or the spectral norm of a matrix, and $\otimes$ be the Kronecker product.
An undirected graph, denoted as $\mathcal{G}=(\mathcal{V},\mathcal{E})$, consists of a vertex set $\mathcal{V}$ and an edge set $\mathcal{E}\subseteq \mathcal{V} \times \mathcal{V}$.
Let $n=|\V|$ and $m=|\E|$.
The set of neighbors of vertex $i$ is denoted as $\mathcal{N}_i\triangleq\{j \in \mathcal{V}: \ (i,j)\in \mathcal{E}\}$.
An {orientation} of an undirected graph is the assignment of a direction to each edge.
An {oriented graph} is an undirected graph together with an orientation.
The {incidence matrix} of an oriented graph is denoted as $H\in\mathbb{R}^{m\times n}$.
For a connected graph, one always has $H\one = 0$ and $\rank(H)=n-1$.

\section{Preliminaries to Bearing Rigidity Theory}\label{section_preliminary}

Bearing rigidity theory will play a key role in the analysis of bearing-based distributed control and estimation problems.
In this section, we revisit a number of important notions and conclusions in bearing rigidity theory \cite{zhao2014TACBearing}.

We first introduce a particularly important orthogonal projection matrix operator.
For any nonzero vector $x\in\R^d$ ($d\ge2$), define the operator $P: \R^d\rightarrow\R^{d\times d}$ as
\begin{align*}
    P(x) \triangleq I_d - \frac{x}{\|x\|}\frac{x^\T }{\|x\|}.
\end{align*}
For notational simplicity, we denote $P_x=P(x)$.  Note that $P_x$ is an orthogonal projection matrix that geometrically projects any vector onto the orthogonal compliment of $x$.
It is easily verified that $P_x^\T =P_x$, $P_x^2=P_x$, and $P_x$ is positive semi-definite.
Moreover, $\Null(P_x)=\myspan\{x\}$ and $P_x$ has one zero eigenvalue and $d-1$ eigenvalues as $1$.

Given a finite collection of $n$ points $\{p_i\}_{i=1}^n$ in $\R^d$ ($n\ge2$, $d\ge2$), a \emph{configuration} is denoted as $p=[p_1^\T,\dots,p_n^\T]^\T\in\mathbb{R}^{dn}$.
A \emph{framework} in $\R^d$, denoted as $\G(p)$, is an undirected graph $\G=(\V,\E)$ together with a configuration $p$, where vertex $i\in\V$ in the graph is mapped to to the point $p_i$ in the configuration.
For a framework $\G(p)$, define the \emph{edge vector} and the \emph{bearing}, respectively, as
\begin{align*}
e_{ij}\triangleq p_j-p_i, \quad g_{ij}\triangleq e_{ij}/\|e_{ij}\|, \quad \forall(i,j)\in\E.
\end{align*}
The bearing $g_{ij}$ is a unit vector.
Note $e_{ij}=-e_{ji}$ and $g_{ij}=-g_{ji}$.
It is often helpful to consider an oriented graph and express the edge vector and the bearing for the $k$th directed edge of the oriented graph as $e_{k}\triangleq p_j-p_i$, $g_{k}\triangleq {e_{k}}/{\|e_{k}\|}$, $\forall k\in\{1,\dots,m\}$.
Let $e=[e_1^\T ,\dots,e_m^\T ]^\T$ and $g=[g_1^\T ,\dots,g_m^\T ]^\T$.
Note $e$ satisfies $e=\bar{H}p$ where $\bar{H}=H\otimes I_d$ and $H$ is the incidence matrix.
Define the \emph{bearing function} $F_B: \R^{dn}\rightarrow\R^{dm}$ as
\begin{align*}
    F_B(p)\triangleq [g_1^\T, \dots,g_m^\T]^\T.
\end{align*}
The bearing function describes all the bearings in the framework.
The \emph{bearing rigidity matrix} is defined as the Jacobian of the bearing function,
\begin{align}\label{eq_rigidityMatrixDefinition}
    R_B(p) \triangleq \frac{\partial F_B(p)}{\partial p}\in\R^{dm\times dn}.
\end{align}
Two important properties of the bearing rigidity matrix are given as below.

\begin{lemma}[\cite{zhao2014TACBearing}]
The bearing rigidity matrix in \eqref{eq_rigidityMatrixDefinition} can be expressed as
$R_B(p)= \mydiag\left({P_{g_k}}/{\|e_k\|}\right)\bar{H}$.
\end{lemma}
\begin{lemma}[\cite{zhao2014TACBearing}]\label{lemma_bearingRigidityMatrixRank}
For any framework $\G(p)$, the bearing rigidity matrix satisfies $\rank(R_B)\le dn-d-1$ and $\myspan{\one\otimes I_d, p}\subseteq \Null(R_B)$.
\end{lemma}

Let $\delta p$ be a variation of $p$.
If $R_B(p)\delta p=0$, then $\delta p$ is called an \emph{infinitesimal bearing motion} of $\G(p)$.
A motion is an infinitesimal bearing motion if and only if the motion preserves the bearing between any pair of neighbors in the framework.
An arbitrary framework always has two kinds of \emph{trivial} infinitesimal bearing motions: translation and scaling of the entire framework.
We next define one of the most important concepts in bearing rigidity theory.

\begin{definition}[{Infinitesimal Bearing Rigidity}]\label{definition_infinitesimalParallelRigid}
    A framework is \emph{infinitesimally bearing rigid} if the infinitesimal bearing motions of the framework are trivial.
\end{definition}

The following is a necessary and sufficient condition for infinitesimal bearing rigidity.

\begin{theorem}[\cite{zhao2014TACBearing}]\label{theorem_conditionInfiParaRigid}
    For any framework $\G(p)$, the following statements are equivalent:
    \begin{enumerate}[(a)]
    \item $\G(p)$ is {infinitesimally bearing rigid};
    \item $\rank(R_B)=dn-d-1$;
    \item $\Null(R_B)=\myspan\{\one\otimes I_d, p\}$.
    \end{enumerate}
\end{theorem}

\begin{theorem}[\cite{zhao2014TACBearing}]\label{theorem_IBRImplyUniqueShape}
An infinitesimally bearing rigid framework can be uniquely determined by the inter-neighbor bearings up to a translation and a scaling factor.
\end{theorem}

As will be shown later, infinitesimal bearing rigidity plays an important role in bearing-only network localization.

\section{Problem Formulation}\label{section_problemFormulation}

In this section, we present the formulations of the two problems of bearing-based formation control and bearing-based network localization, and then propose linear controllers and estimators to solve them.

Consider a network of $n$ agents in $\R^d$ ($n\ge2$, $d\ge2$).
The network may represent a multi-vehicle system or a sensor network.
Denote ${p}_i\in\R^d$ as the position of agent $i\in\{1,\dots,n\}$, and let $p=[p_1^\T ,\dots,p_n^\T ]^\T \in\R^{dn}$.
The interaction graph $\G=(\V,\E)$ is assumed to be connected, undirected, and fixed.
The network is denoted as $\G(p)$.

\subsection{Bearing-Based Formation Control}

In this subsection, we study the problem of bearing-based formation control, the aim of which is to stabilize a target formation with bearing constraints using relative position measurements.
Specifically, the target formation is specified by the constant {bearing constraints} $\{g_{ij}^*\}_{(i,j)\in\E}$ where $g_{ij}^*=-g_{ji}^*$.
The bearing constraints must be feasible such that there exist formations satisfying the constraints.
Suppose agent $i$ can measure the relative positions of its neighbors, $\{p_i-p_{j}\}_{j\in\N_i}$.
The dynamics of agent $i$ are assumed to be a single integrator $\dot{p}_i(t) = u_i(t)$, where $u_i(t)\in\R^d$ is the input to be designed.
The problem of bearing-based formation control is stated as below.

\begin{problem}[Bearing-Based Formation Control]\label{problem_bearingonlyformationcontrol}
    Given feasible constant bearing constraints $\{g_{ij}^*\}_{(i,j)\in\E}$ and an initial position $p(0)$, design $u_i(t)$ ($i\in\V$) based on the relative position measurements $\{p_i(t)-p_j(t)\}_{j\in\N_i}$ such that $g_{ij}(t)\rightarrow g_{ij}^*$ as $t\rightarrow\infty$ for all $(i,j)\in\E$.
\end{problem}

We next propose two controllers to solve Problem~\ref{problem_bearingonlyformationcontrol}.
The first controller is leaderless, and the second one assumes fixed leaders in the formation.

\subsubsection{Leaderless Case}
The leaderless formation controller is designed as
\begin{align}\label{eq_formationControlLaw_leaderless}
    \dot{p}_i(t)=-\sum_{j\in\N_i} P_{g_{ij}^*} (p_i(t)-p_j(t)),\quad i\in\V,
\end{align}
where $P_{g_{ij}^*}=I_d-g_{ij}^*(g_{ij}^*)^\T $.
The matrix expression of controller \eqref{eq_formationControlLaw_leaderless} is
\begin{align}\label{eq_formationControlLawMatrix_leaderless}
    \dot{p}(t)=-\L(\G,g^*) p(t),
\end{align}
where $\L(\G,g^*)\in\R^{dn\times dn}$ is a matrix-weighted Laplacian.
In particular, let $[\L(\G,g^*)]_{ij}$ be the $ij$th block of size $d\times d$ in $\L(\G,g^*)$, then
\begin{align}\label{eq_laplacian_projectionMatrix}
    [\L(\G,g^*)]_{ij}&=-P_{g_{ij}^*}, \quad j\ne i, \nonumber\\
    [\L(\G,g^*)]_{ii}&=\sum_{j\in\N_i}P_{g_{ij}^*}, \quad i\in\V.
\end{align}
When the context is clear, we simply write $\L(\G,g^*)$ as $\L$.

\subsubsection{Leader-Follower Case}
As will be shown later, controller~\eqref{eq_formationControlLawMatrix_leaderless} can successfully solve Problem~\ref{problem_bearingonlyformationcontrol} but cannot the centroid or the scale of the formation.
Motivated by this, we introduce fixed leaders and propose a leader-follower controller.
Suppose there are $n_l$ ($0\le n_l \le n$) fixed agents which are called \emph{leaders}.
The rest $(n-n_l)$ agents are called \emph{followers}.
Denote $\V_l$ and $\V_f$ as the index sets for the leaders and followers, respectively.
Note $\V_l \cup \V_f = \V$, $|\V_l|=n_l$, and $|\V_f|=n-n_l$.
When $n_l=0$, it will be the same as the above leaderless case.
The leader-follower formation controller is designed as
\begin{align}\label{eq_formationControlLaw_leaderFollower}
    \dot{p}_i(t)&=0, \quad i\in\V_l, \nonumber\\
    \dot{p}_i(t)&=-\sum_{j\in\N_i} P_{g_{ij}^*} (p_i(t)-p_j(t)),\quad i\in\V_f.
\end{align}
Assume without loss of generality that the first $n_l$ agents are leaders and the rest are followers.
As a result, $\V_l=\{1,\dots,n_l\}$ and $\V_f=\{n_l+1,\dots,n\}$.
Denote $p_l=[p_1^\T ,\dots,p_{n_l}^\T ]^\T \in\R^{dn_l}$ and $p_f=[p_{n_l+1}^\T ,\dots,p_{n}^\T ]^\T \in\R^{d(n-n_l)}$.
Partition $\L$ into the following form
\begin{align*}
    \L=\left[
         \begin{array}{cc}
           \L_{ll} & \L_{lf} \\
           \L_{fl} & \L_{ff} \\
         \end{array}
       \right],
\end{align*}
where $\L_{ll}\in\R^{dn_l\times dn_l}$, $\L_{lf}=\L_{fl}^\T \in\R^{dn_l\times d(n-n_l)}$, and $\L_{ff}\in\R^{d(n-n_l)\times d(n-n_l)}$.
Then, the matrix expression of controller \eqref{eq_formationControlLaw_leaderFollower} is
\begin{align}\label{eq_formationControlLawMatrix_leaderFollower}
    \dot{p}_l(t)&=0, \nonumber \\
    \dot{p}_f(t)&=-\L_{ff}p_f(t)-\L_{fl}p_l,
\end{align}
where $p_l$ is constant since the leaders are stationary.

\subsection{Bearing-Based Network Localization}

In this subsection, we consider the problem of bearing-based network localization, the aim of which is to estimate the positions of the agents in a network merely using bearing measurements.
Agent $i$ maintains an estimate $\hat{p}_i$ of its own fixed position $p_i$.
Agent $i$ can measure the bearings of its neighbors, $\{g_{ij}\}_{j\in\N_i}$, and can also obtain the estimates of its neighbors (via communication), $\{\hat{p}_j\}_{j\in\N_i}$.
The estimation update law for agent $i$ has the form of $\dot{\hat{p}}_i(t) = u_i(t)$, where $u_i(t)\in\R^d$ is the input to be designed.
The bearing-based network localization problem is stated as below.

\begin{problem}[Bearing-Based Network Localization]\label{problem_bearingbasednetworkLocalization}
    Suppose $p$ is constant. Given an initial estimate $\hat{p}(0)$, design $u_i(t)$ ($i\in\V$) based on the relative estimates $\{\hat{p}_i(t)-\hat{p}_j(t)\}_{j\in\N_i}$ and the constant bearing measurements $\{g_{ij}\}_{j\in\N_i}$ such that $\hat{p}_i(t)\rightarrow p_i$ as $t\rightarrow\infty$ for all $i\in\V$.
\end{problem}

Suppose a subset of the agents, known as \emph{anchors}, can measure their own real positions; the remaining agents are called \emph{followers}.
We do not consider the anchorless case here because the network cannot be localized without anchors.
Suppose there are $n_a$ ($0\le n_a \le n$) anchors and $n-n_a$ followers.
Denote $\V_a$ and $\V_f$ as the index sets for the anchors and followers, respectively.
Then $\V_a \cup \V_f = \V$, $|\V_a|=n_a$, and $|\V_f|=n-n_a$.
The anchor-follower estimator is designed as
\begin{align}\label{eq_networkLocalize_achor}
    \dot{p}_i(t)&=0, \quad i\in\V_a, \nonumber\\
    \dot{p}_i(t)&=-\sum_{j\in\N_i} P_{g_{ij}} (\hat{p}_i(t)-\hat{p}_j(t)),\quad i\in\V_f.
\end{align}
It is notable that the above estimator has the {same} formula as the controller~\eqref{eq_formationControlLaw_leaderFollower}.
Define $\L(\G,g)\in\R^{dn\times dn}$ with
\begin{align}\label{eq_laplacian_projectionMatrix_netEst}
    [\L(\G,g)]_{ij}&=-P_{g_{ij}}, \quad j\ne i, \nonumber\\
    [\L(\G,g)]_{ii}&=\sum_{j\in\N_i}P_{g_{ij}}, \quad i\in\V.
\end{align}
When the context is clear, we simply write $\L(\G,g)$ as $\L$.
Without loss of generality, assume the first $n_a$ agents are anchors and the others are followers.
As a result, $\V_a=\{1,\dots,n_a\}$ and $\V_f=\{n_a+1,\dots,n\}$.
Denote $p_a=[p_1^\T ,\dots,p_{n_a}^\T ]^\T \in\R^{dn_a}$ and $p_f=[p_{n_a+1}^\T ,\dots,p_{n}^\T ]^\T \in\R^{d(n-n_a)}$.
Partition $\L$ into the following form
\begin{align*}
    \L=\left[
         \begin{array}{cc}
           \L_{aa} & \L_{af} \\
           \L_{fa} & \L_{ff} \\
         \end{array}
       \right],
\end{align*}
where $\L_{aa}\in\R^{dn_a\times dn_a}$, $\L_{af}=\L_{fa}^\T \in\R^{dn_a\times d(n-n_a)}$, and $\L_{ff}\in\R^{d(n-n_a)\times d(n-n_a)}$.
Then, it is straightforward to see the matrix expression of controller \eqref{eq_networkLocalize_achor} is
\begin{align}\label{eq_networkLocalizeMatrix_achor}
    \dot{\hat{p}}_a(t)&=0, \nonumber\\
    \dot{\hat{p}}_f(t)&=-\L_{ff}\hat{p}_f(t)-\L_{fa}p_a,
\end{align}
where $p_a$ is constant and known.

\section{Convergence Analysis}\label{section_convergenceAnalysis}

In this section, we analyze the convergence of the proposed controllers and estimators.
Since the network localization has the same form as the formation control, we will mainly focus on the convergence analysis of the formation control and the convergence results for the network localization are given without proofs.

\subsection{Convergence Analysis of Formation Control}

In order to prove the convergence of the proposed controllers, we adopt the following assumption.

\begin{assumption}\label{assumption_IBR}
    Any formation $\G(p^*)$ that satisfies the bearing constraints $\{g_{ij}^*\}_{(i,j)\in\E}$ is infinitesimally bearing rigid.
\end{assumption}

Assumption~\ref{assumption_IBR} gives two useful conditions.
By Theorem~\ref{theorem_IBRImplyUniqueShape}, the first condition is that the target formation specified by the bearing constraints has a unique shape.
By Theorem~\ref{theorem_conditionInfiParaRigid}, the second condition is a mathematical condition that $\rank(R_B(p^*))=dn-d-1$ and $\Null(R_B(p^*))=\myspan\{\one\otimes I_d, p^*\}$ where $R_B(p^*)=\dia{P_{g_k^*}/\|e_k^*\|}\bar{H}$ is the bearing rigidity matrix.
Since the distance term $\|e_k^*\|$ in $R_B(p^*)$ does not affect its rank or null space, the condition given by Assumption~\ref{assumption_IBR} actually is
\begin{align*}
    \Null(\dia{P_{g_k^*}}\bar{H})=\myspan\{\one\otimes I_d, p^*\}.
\end{align*}
This condition will be crucial to the following convergence analysis.

\subsubsection{Leaderless Case}

We first consider the case without leaders and study the formation dynamics \eqref{eq_formationControlLawMatrix_leaderFollower}.
Since the system \eqref{eq_formationControlLawMatrix_leaderFollower} is linear and time-invariant, its convergence is totally determined by the spectrum of $\L$.
The next result characterizes the rank and null space of $\L$.

\begin{lemma}\label{lemma_FormationLaplacian_nullspace}
    Under Assumption~\ref{assumption_IBR}, the $\L$ defined in \eqref{eq_laplacian_projectionMatrix} is symmetric positive semi-definite. Moreover, it satisfies $\rank(\L)=dn-d-1$ and
    \begin{align*}
        \Null(\L)=\myspan\{\one\otimes I_d, p^*\},
    \end{align*}
    where $p^*$ is an arbitrary configuration satisfying $\{g_{ij}^*\}_{(i,j)\in\E}$.
\end{lemma}
\begin{proof}
    Since the graph is undirected, $\L$ can be written as $\L=\bar{H}^\T \mydiag(P_{g_k^*})\bar{H}$, which is clearly symmetric and positive semi-definite.
    Moreover, due to $P_{g_k^*}=P_{g_k^*}^\T P_{g_k^*}$, $\L$ can be rewritten as
    \begin{align*}
        \L=\underbrace{\bar{H}^\T \mydiag{(P_{g_k^*}^\T )}}_{\tilde{R}_B^\T }\underbrace{\mydiag{(P_{g_k})}\bar{H}}_{\tilde{R}_B}.
    \end{align*}
    By Assumption~\ref{assumption_IBR}, we have $\rank(\tilde{R}_B)=dn-d-1$ and $\Null(\tilde{R}_B)=\myspan\{\one\otimes I_d, p^*\}$.
    Since $\L$ has the same rank and null space as $\tilde{R}_B$, the proof is complete.
\end{proof}

Based on Lemma~\ref{lemma_FormationLaplacian_nullspace}, we can analyze the convergence of system \eqref{eq_formationControlLawMatrix_leaderless}.
Since $\L$ is symmetric, its left and right null spaces are the same.
Although $\{\one\otimes I_d, p^*\}$ is a basis of $\Null(\L)$ by Lemma~\ref{lemma_FormationLaplacian_nullspace}, it is not an orthogonal basis in general.
In order to obtain an orthogonal basis, we first define the \emph{formation centroid} (denoted $c(p)$), \emph{normalized formation} (denoted $r(p)$), and \emph{formation scale} (denoted $s(p)$) as
\begin{align*}
c(p)\triangleq \frac{\one^\T  p}{n}, \quad
r(p)\triangleq p-\one\otimes c(p), \quad
s(p)\triangleq\|r(p)\|.
\end{align*}
Note $r(p)$ is always orthogonal to $\one\otimes I_d$.
As a result, if denoting $r^*=p^*-\one\otimes c(p^*)$, we have $\{\one\otimes I_d, r^*\}$ is an orthogonal basis of $\Null(\L)$.
When the context is clear, we will simply write $c(p), r(p), s(p)$ as $c, r, s$.

\begin{theorem}[Convergence of Leaderless Control]\label{theorem_bearingFormation_leaderless_convergence}
Under Assumption~\ref{assumption_IBR}, the trajectory of system~\eqref{eq_formationControlLawMatrix_leaderless} converges exponentially from any initial point $p(0)$ to
\begin{align*}
p(\infty)= \one\otimes c(0)+ \left(\frac{(r^*)^\T }{\|r^*\|}p(0)\right) \frac{r^*}{\|r^*\|}.
\end{align*}
If $(r^*)^\T p(0)>0$, the leaderless controller~\eqref{eq_formationControlLaw_leaderless} successfully solves Problem~\ref{problem_bearingonlyformationcontrol}.
Furthermore, the centroid and the scale of the final formation are $c(\infty)=c(0)$ and $s(\infty)=\left|(r^*)^\T p(0)/\|r^*\|\right|$, respectively.
\end{theorem}
\begin{proof}
Denote $A=[\one\otimes I_d,r^*]\in\R^{dn\times(d+1)}$. By the linear system theory, the trajectory $p(t)$ of system~\eqref{eq_formationControlLawMatrix_leaderless} converges to the orthogonal projection of $p(0)$ onto $\Range(A)$:
\begin{align*}
p(\infty)
&=A(A^\T A)^{-1}A^\T p(0) \\
&= (\one\otimes I_d)\left((\one\otimes I_d)^\T (\one\otimes I_d)\right)^{-1}(\one\otimes I_d)^\T  p(0) \\
&\qquad + \frac{r^*(r^*)^\T }{(r^*)^\T r^*}p(0) \\
&= \one\otimes c(0)+ \frac{(r^*)^\T p(0)}{\|r^*\|} \frac{r^*}{\|r^*\|}.
\end{align*}
It is easy to verify that the centroid and the scale of the final formation are $c(\infty)=c(0)$ and $s(\infty)=\left|(r^*)^\T p(0)/\|r^*\|\right|$, respectively.
The final formation $p(\infty)$ can be obtained by translating and scaling $r^*$.
Since $\G(r^*)$ satisfies all the bearing constraints, $\G(p(\infty))$ also satisfies as long as the scaling factor $(r^*)^\T p(0)/\|r^*\|$ is positive.
\end{proof}

Several remarks regarding Theorem~\ref{theorem_bearingFormation_leaderless_convergence} are given here.
\begin{enumerate}[(a)]
\item When $(r^*)^\T p(0)<0$, the formation converges to a final formation with the bearings as $g_{ij}=-g^*_{ij},\forall (i,j)\in\E$ instead of $g_{ij}=g^*_{ij},\forall (i,j)\in\E$. In this case, although the final formation has the opposite bearings as desired, it can be viewed as a point reflection of the target formation and has the same shape.
\item The centroid of the final formation is the same as that of the initial formation.
In fact, it follows from $(\one\otimes I_d)^\T  \dot{p}=0$ that the centroid of the formation is invariant under controller~\eqref{eq_formationControlLaw_leaderless}.
\item Although the centroid is invariant, the scale of the formation is changed under controller~\eqref{eq_formationControlLaw_leaderless}.
Specifically, the scale of the final formation satisfies
\begin{align*}
0\le s(\infty) \le s(0).
\end{align*}
The scale of the final formation is no larger than that of the initial formation.
It is clear that the lower bound of $s(\infty)$ is achieved when $(r^*)^\T p(0)=0$.
In this case, the formation will finally reach rendezvous (i.e., consensus in terms of position).
In order to obtain the upper bound, rewrite $(r^*)^\T p(0)=(r^*)^\T (p(0)-\one\otimes c(p(0)))=(r^*)^\T  r(0)$.
Then $s(\infty)=|(r^*)^\T r(0)/\|r^*\||\le\|r(0)\|=s(0)$.
As a result, the upper bound of $s(\infty)$ is achieved when $r^*$ is parallel to $p(0)$ or $r(0)$.
\end{enumerate}

Fig.~\ref{fig_demoControlLaw} shows the simplest example to demonstrate controller~\eqref{eq_formationControlLaw_leaderless}.
In this example, the target formation is vertical.
The final position of each agent is the orthogonal projection of its initial position to the target bearing.
Moreover, the centroid of the formation is invariant but the scale is changed.
It is intuitively obvious that if the initial formation is horizontal, the two agents will reach rendezvous.

\begin{figure}
  \centering
  \includegraphics[width=0.4\linewidth]{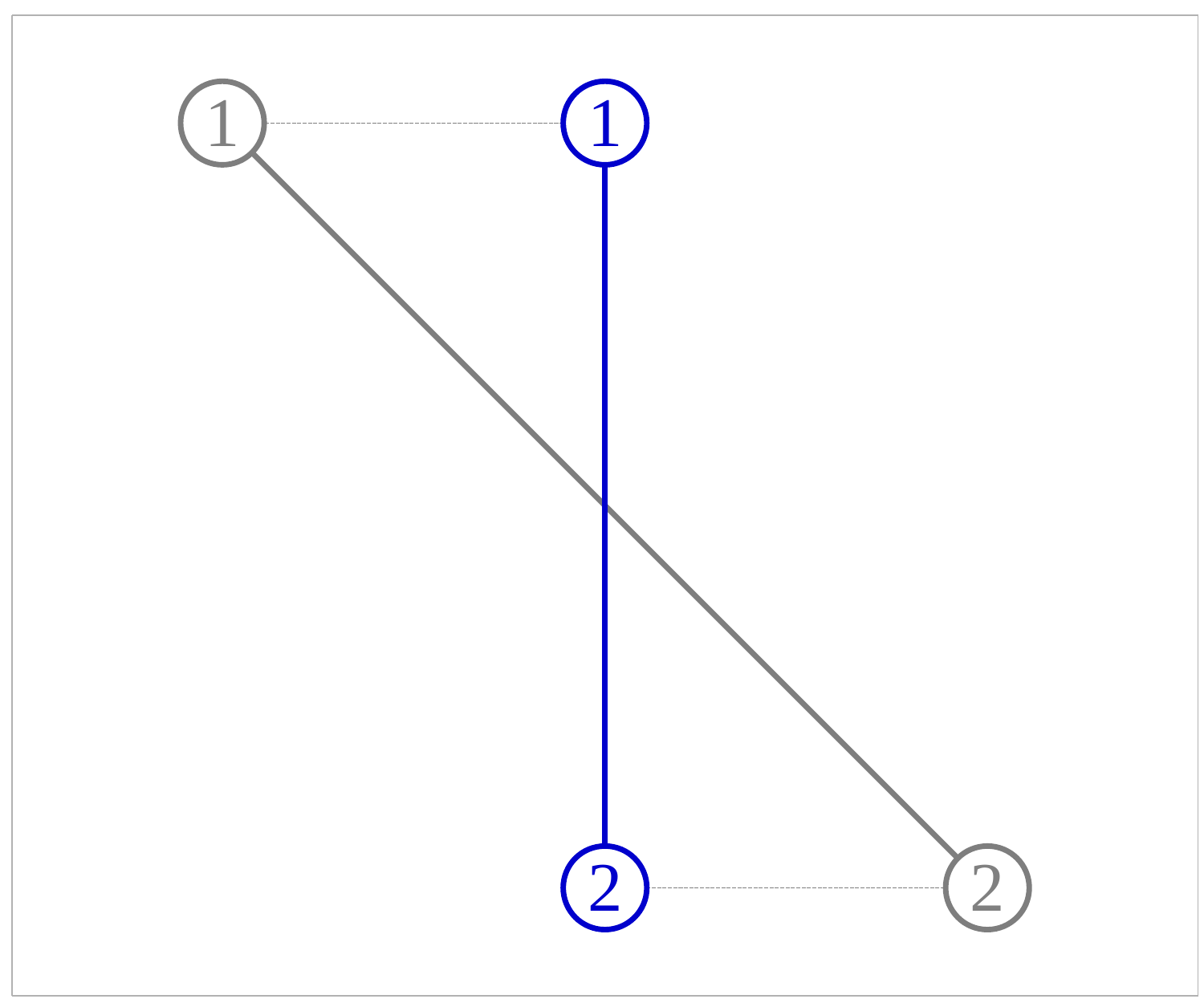}
  \caption{The simplest example to demonstrate the geometric interpretation of the leaderless controller~\eqref{eq_formationControlLaw_leaderless}. Initial formation: gray; target/final formation: blue; agent trajectory: dotted line.}
  \label{fig_demoControlLaw}
\end{figure}

\subsubsection{Leader-Follower Case}

Although the controller~\eqref{eq_formationControlLaw_leaderless} is able to solve Problem~\ref{problem_bearingonlyformationcontrol}, the centroid and the scale of the finally converged formation are determined by the initial formation, which is usually undesired in practice.
We next show that the centroid and the scale of the formation can be controlled by the leader-follower controller~\eqref{eq_formationControlLaw_leaderFollower}.
We first analyze the properties of $\L_{ff}$ in the leader-follower controller.

\begin{lemma}\label{lemma_formationControl_LffPositiveDefinite}
    Under Assumption~\ref{assumption_IBR}, $\L_{ff}$ in system~\eqref{eq_formationControlLaw_leaderFollower} is positive definite if and only if $n_l\ge2$.
\end{lemma}
\begin{proof}
For any $x\in\R^{dn_l}$, since $\L\ge0$, we have
\begin{align*}
x^\T  \L_{ff} x =\left[
                  \begin{array}{cc}
                    0 & x^\T  \\
                  \end{array}
                \right]
                \left[
                  \begin{array}{cc}
                    \L_{ll} & \L_{lf} \\
                    \L_{fl} & \L_{ff} \\
                  \end{array}
                \right]
                \left[
                  \begin{array}{c}
                    0 \\
                    x \\
                  \end{array}
                \right]\ge0.
\end{align*}
As a result, $\L_{ff}$ is at least positive semi-definite.
If there exists a nonzero vector $x$ such that $x^\T L_{ff}x=0$, then $[0,x^\T ]^\T \in\Null(\L)=\myspan\{\one\otimes I_d, r^*\}$.
If there is only one leader ($n_l=1$), it is easy to see that such $x$ exists.
However, if there are more than one leaders ($n_l\ge2$), such $x$ does not exist because $p_i\ne p_j$ for all $i\ne j$.
Thus, in the case of $n_l\ge2$, $L_{ff}$ is positive definite.
\end{proof}

When $n_l\ge2$, the positions of the leaders, $p_l$, must be \emph{feasible} such that the followers together with the leaders can possibly form a formation satisfying the bearing constraints.
The following is a necessary condition for a feasible $p_l$.

\begin{lemma}\label{lemma_formationControl_feasibleLeader}
    Under Assumption~\ref{assumption_IBR}, a feasible $p_l$ satisfies
    \begin{align*}
        \left(\L_{ll}-\L_{lf}\L_{ff}^{-1}\L_{fl}\right) p_l=0.
    \end{align*}
\end{lemma}
\begin{proof}
If $p_l$ is feasible, there exists $p_f$ such that $p=[p_l^\T ,p_f^\T ]^\T $ satisfies the bearing constraints $\{g_{ij}^*\}_{(i,j)\in\E}$.
By Theorem~\ref{theorem_IBRImplyUniqueShape}, infinitesimal bearing rigidity can uniquely determine $p$ up to a translation and a scaling factor.
That means $p\in\myspan\{\one\otimes I_d, r^*\}=\Null(\L)$ and consequently
\begin{align*}
\left[
  \begin{array}{cc}
    \L_{ll} & \L_{lf} \\
    \L_{fl} & \L_{ff} \\
  \end{array}
\right]
\left[
  \begin{array}{c}
    p_l \\
    p_f \\
  \end{array}
\right]=0,
\end{align*}
which implies $\L_{ll}p_l+\L_{lf}p_f=0$ and $\L_{fl}p_l+\L_{ff}p_f=0$.
The second equation implies $p_f=-\L_{ff}^{-1}\L_{fl}p_l$, substituting which into the first equation completes the proof.
\end{proof}

Based on Lemmas~\ref{lemma_formationControl_LffPositiveDefinite} and \ref{lemma_formationControl_feasibleLeader}, we have the following convergence result for the leader-follower formation controller.

\begin{theorem}[Convergence of Leader-Follower Control]\label{theorem_FormationControl_leaderfollower}
    Under Assumption~\ref{assumption_IBR}, given $n_l\ge2$ and a feasible $p_l$, the trajectory of system~\eqref{eq_formationControlLawMatrix_leaderFollower} converges exponentially fast from any initial $p_f(0)$ to
    \begin{align*}
        p_f(\infty)=-\L_{ff}^{-1}\L_{fl}p_l.
    \end{align*}
    The finally converged formation satisfies the bearing constraints $\{g_{ij}^*\}_{(i,j)\in\E}$.
\end{theorem}

\begin{proof}
Since $\L_{ff}>0$ if $n_l\ge2$ by Lemma~\ref{lemma_formationControl_LffPositiveDefinite}, it is obvious that the linear time-invariant system~\eqref{eq_formationControlLawMatrix_leaderFollower} is exponentially stable.
Then the final formation (i.e., the equilibrium) satisfying $\dot{p}_f=0$ is $p_f(\infty)=-\L_{ff}^{-1}\L_{fl}p_l$.
Since $p_l$ is feasible, it follows from Lemma~\ref{lemma_formationControl_feasibleLeader} that $\L p(\infty)=0$ where $p(\infty)=[p_l^\T ,p_f^\T (\infty)]^\T $.
Therefore, $\G(p(\infty))$ satisfies the bearing constraints.
\end{proof}

As shown in Theorem~\ref{theorem_FormationControl_leaderfollower}, the final formation $p_f(\infty)$ is a function of $p_l$.
As a result, we can control the centroid and the scale of the final formation $p(\infty)$ by choosing appropriate positions of the leaders $p_l$.

\subsection{Convergence Analysis of Network Localization}

We adopt the following assumption to analyze the convergence of system~\eqref{eq_networkLocalizeMatrix_achor}.

\begin{assumption}\label{assumption_IBR_netEst}
    The network $\G(p)$ is infinitesimally bearing rigid.
\end{assumption}

The properties of $\L$ defined in \eqref{eq_laplacian_projectionMatrix_netEst} are given below.

\begin{lemma}\label{lemma_netEst_Laplacian_nullspace}
    Under Assumption~\ref{assumption_IBR_netEst}, the $\L$ defined in \eqref{eq_laplacian_projectionMatrix_netEst} is symmetric positive semi-definite. Moreover, it satisfies $\rank(\L)=dn-d-1$ and
    \begin{align*}
        \Null(\L)=\myspan\{\one\otimes I_d, p\}.
    \end{align*}
\end{lemma}
\begin{proof}
Similar to Lemma~\ref{lemma_FormationLaplacian_nullspace}.
\end{proof}

\begin{lemma}\label{lemma_NetEst_LffPositiveDefinite}
    Under Assumption~\ref{assumption_IBR_netEst}, $\L_{ff}$ in system~\eqref{eq_networkLocalizeMatrix_achor} is positive definite if and only if $n_a\ge2$.
\end{lemma}
\begin{proof}
Similar to Lemma~\ref{lemma_formationControl_LffPositiveDefinite}.
\end{proof}

By the above two lemmas, the positions of the anchors and followers satisfy the following condition.

\begin{lemma}\label{lemma_NetEst_anchorCondition}
    Under Assumption~\ref{assumption_IBR_netEst}, if $n_a\ge2$, then $p_a$ and $p_f$ satisfy
    \begin{align*}
        p_f=-\L_{ff}^{-1}\L_{fa}p_a.
    \end{align*}
\end{lemma}
\begin{proof}
Since $\Null(\L)=\myspan\{\one\otimes I_d, p\}$ by Lemma~\ref{lemma_netEst_Laplacian_nullspace}, we have
\begin{align*}
\left[
  \begin{array}{cc}
    \L_{aa} & \L_{af} \\
    \L_{fa} & \L_{ff} \\
  \end{array}
\right]
\left[
  \begin{array}{c}
    p_a \\
    p_f \\
  \end{array}
\right]=0,
\end{align*}
which implies $\L_{fa}p_a+\L_{ff}p_f=0$.
Since $\L_{ff}$ is nonsingular given $n_a\ge2$, $p_f$ can be solved as $p_f=-\L_{ff}^{-1}\L_{fa}p_a$.
\end{proof}

The next is the main convergence result of the anchor-based estimator~\eqref{eq_networkLocalize_achor}.
\begin{theorem}[Convergence of Anchor-based Localization]
    Under Assumption~\ref{assumption_IBR_netEst}, if $n_a\ge2$, the trajectory of system~\eqref{eq_networkLocalizeMatrix_achor} converges exponentially fast from any initial estimate $\hat{p}_f(0)$ to
    \begin{align*}
        \hat{p}_f(\infty)=-\L_{ff}^{-1}\L_{fa}p_a=p_f.
    \end{align*}
    As a result, the estimator~\eqref{eq_networkLocalize_achor} successfully solves Problem~\ref{problem_bearingbasednetworkLocalization}.
\end{theorem}
\begin{proof}
Since $\L_{ff}>0$ if $n_a\ge2$, it is obvious that the linear time-invariant system~\eqref{eq_networkLocalizeMatrix_achor} is exponentially stable.
The the final estimate (i.e., the equilibrium) that satisfies $\dot{\hat{p}}_f=0$ is given by $\hat{p}_f(\infty)=-\L_{ff}^{-1}\L_{fa}p_a$.
Since $p_f=-\L_{ff}^{-1}\L_{fa}p_a$ according to Lemma~\ref{lemma_NetEst_anchorCondition}, we know $\hat{p}_f(\infty)=p_f$.
\end{proof}

\section{Simulation Examples}\label{section_simulation}

\subsection{Examples for Bearing-Based Formation Control}

We now show two examples in Figs.~\ref{fig_sim_3DCube_leaderless} and \ref{fig_sim_3DCube_leaderFollower} to verify the leaderless and leader-follower controllers~\eqref{eq_formationControlLaw_leaderless} and \eqref{eq_formationControlLaw_leaderFollower}, respectively.
The target formation for the two examples is a three-dimensional cube with 8 agents and 13 edges.
The target formation is infinitesimally bearing rigid because $\rank(R_B)=20=3n-4$ according to Theorem~\ref{theorem_conditionInfiParaRigid}.
The initial formations are randomly generated.
There are no leaders in the example in Fig.~\ref{fig_sim_3DCube_leaderless}, while there are two fixed leaders in Fig.~\ref{fig_sim_3DCube_leaderFollower}.
As can be seen, the target formation can be achieved in both of the two examples.
But the centroid and scale of the final formations are different.
In the leader-follower case, the centroid and the scale of the final formation are determined by the two fixed leaders.

\begin{figure}
  \centering
  \subfloat[Initial formation]{\includegraphics[width=0.4\linewidth]{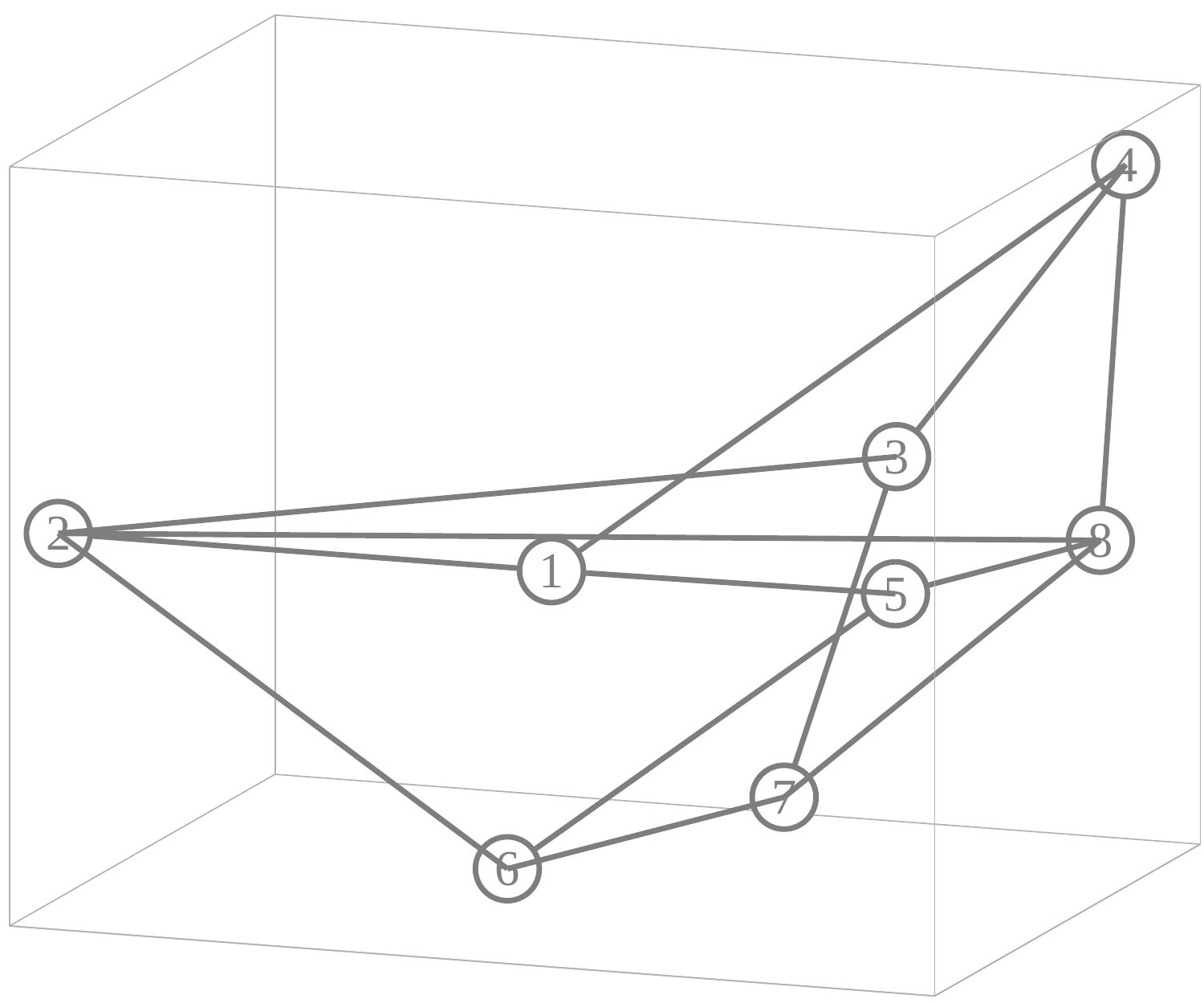}}\,
  \subfloat[Final formation]{\includegraphics[width=0.4\linewidth]{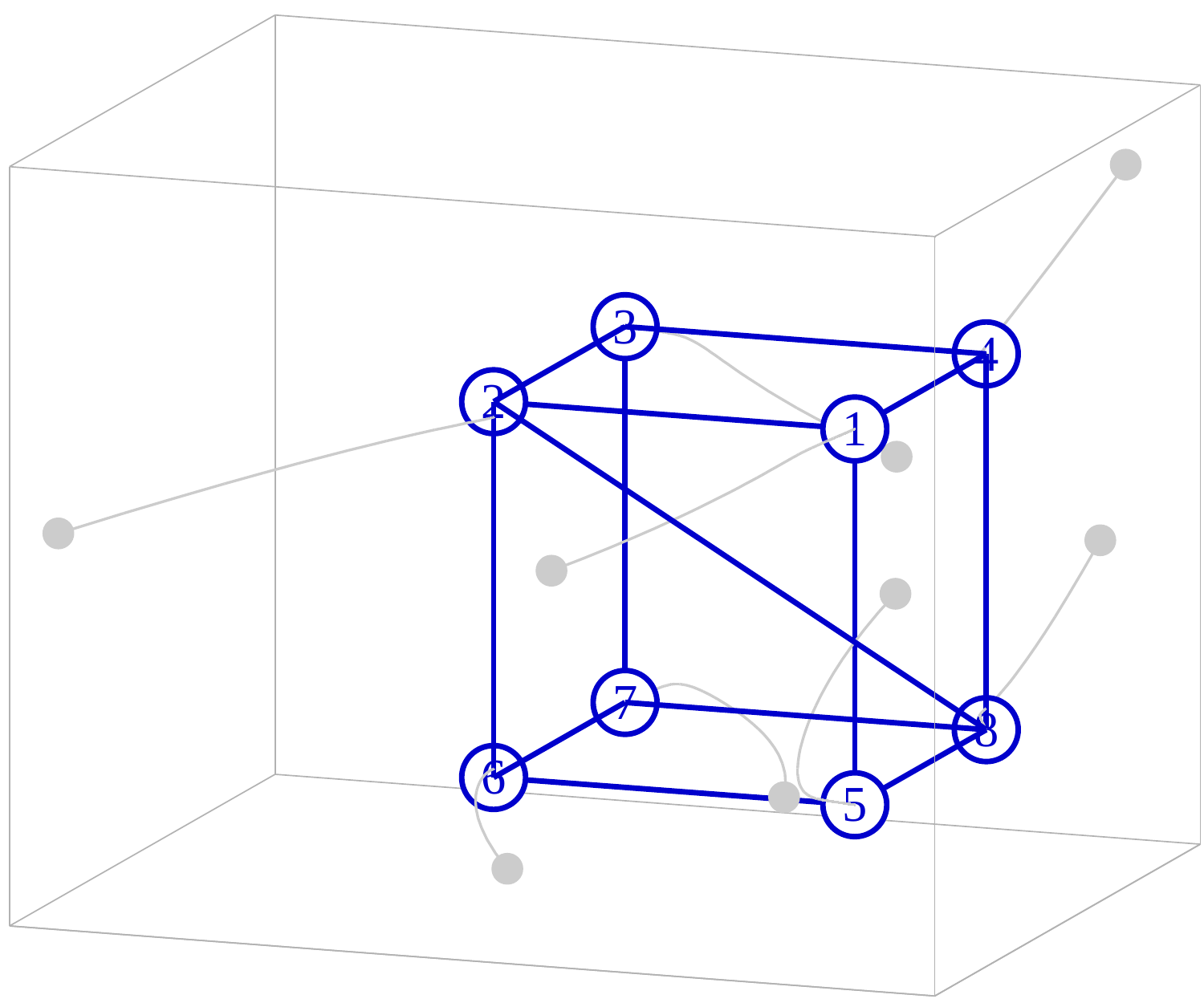}}
  \caption{A 3D example for bearing-based formation control: leaderless case.}
  \label{fig_sim_3DCube_leaderless}
\end{figure}
\begin{figure}
  \centering
  \subfloat[Initial formation]{\includegraphics[width=0.4\linewidth]{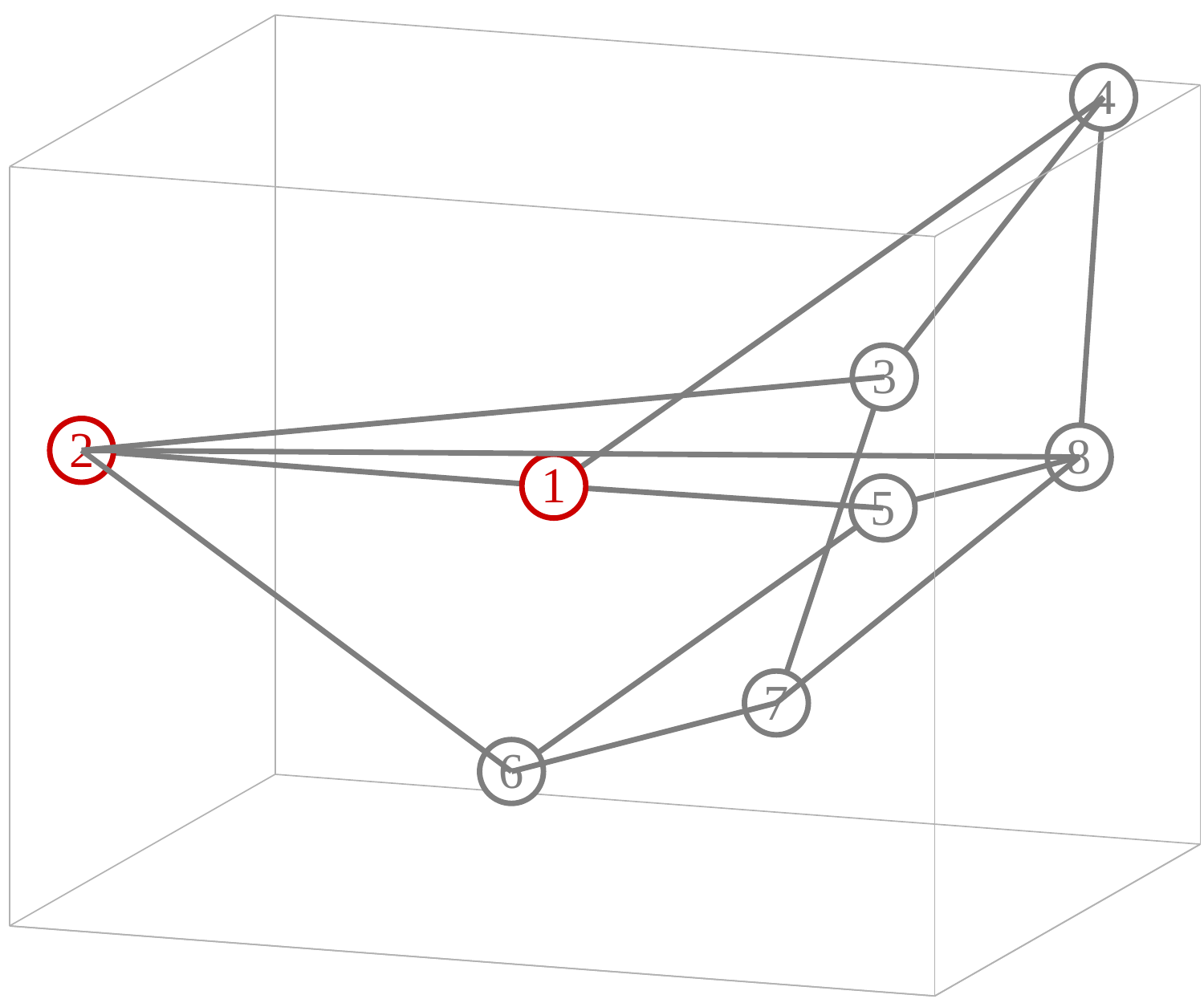}}\,
  \subfloat[Final formation]{\includegraphics[width=0.4\linewidth]{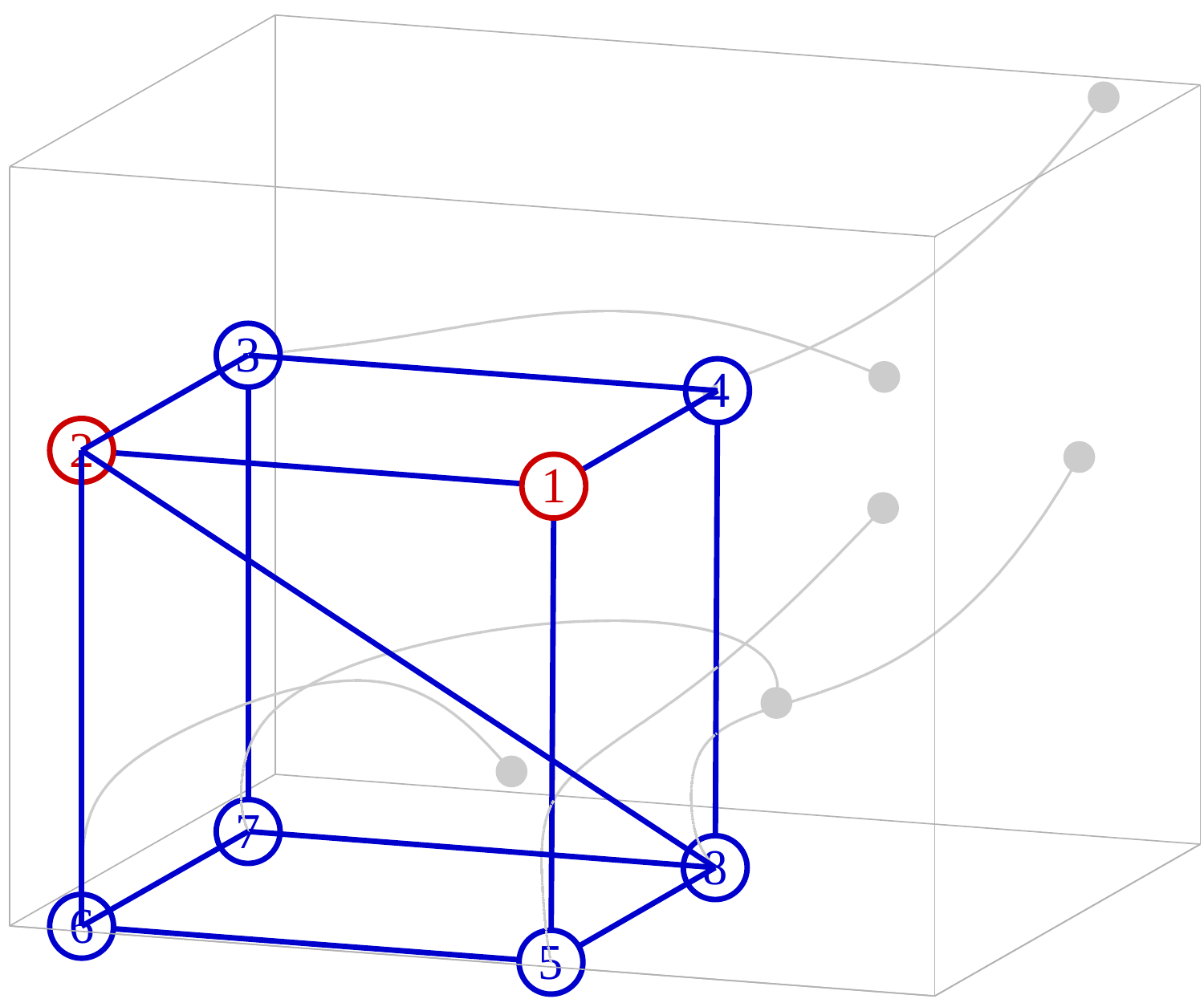}}
  \caption{A 3D example for bearing-based formation control: leader-follower case. Agents in red are fixed leaders.}
  \label{fig_sim_3DCube_leaderFollower}
\end{figure}

\subsection{Examples for Bearing-Based Network Localization}

We next present examples to verify the anchor-based network localization estimator \eqref{eq_networkLocalize_achor}.
Fig.~\ref{fig_sim_NetworkEst}(a) shows a three-dimensional network with 50 agents, 269 edges, and 4 fixed anchors.
The network is infinitesimally bearing rigid because $\rank(R_B)=146=3n-4$.
The initial estimate shown in in Fig.~\ref{fig_sim_NetworkEst}(b) is randomly chosen.
As can be seen in Fig.~\ref{fig_sim_NetworkEst}(c)-(d), the estimate errors $\|\hat{p}_i(t)-p_i\|$ finally converge to zero.

\section{Conclusions}\label{section_conclusion}

In this paper, we applied bearing rigidity theory to solve two bearing-based control and estimation problems in arbitrary dimensional spaces.
The proposed linear controllers and estimators can globally solve the two problems without ambiguity, respectively.
This paper only considered the cases of undirected and fixed underlying graphs.
One may study the cases of directed and switching graphs in the future.

\begin{figure}
  \centering
  \subfloat[Real network]{\includegraphics[width=0.25\linewidth]{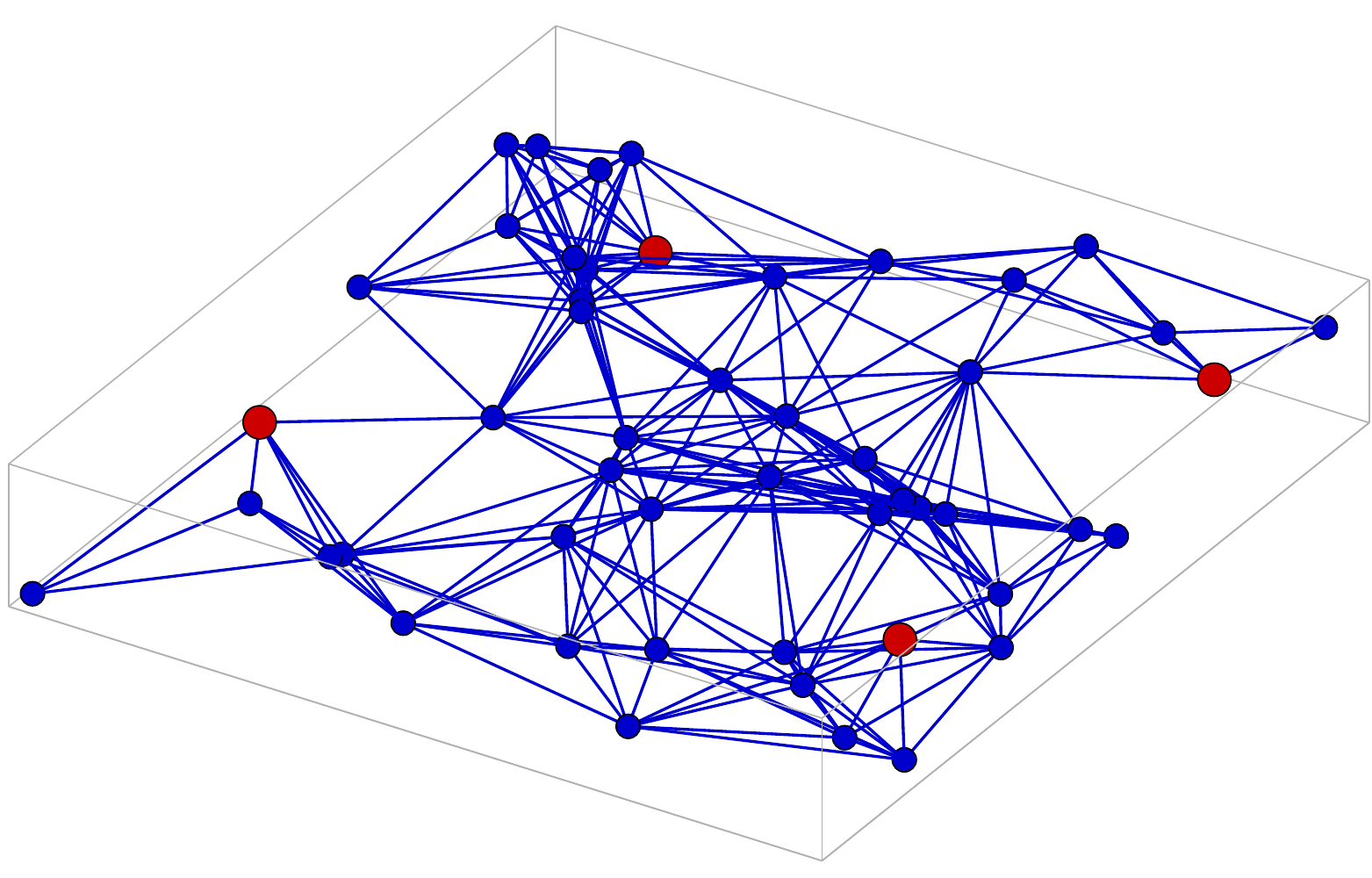}}
  \subfloat[Initial localization]{\includegraphics[width=0.25\linewidth]{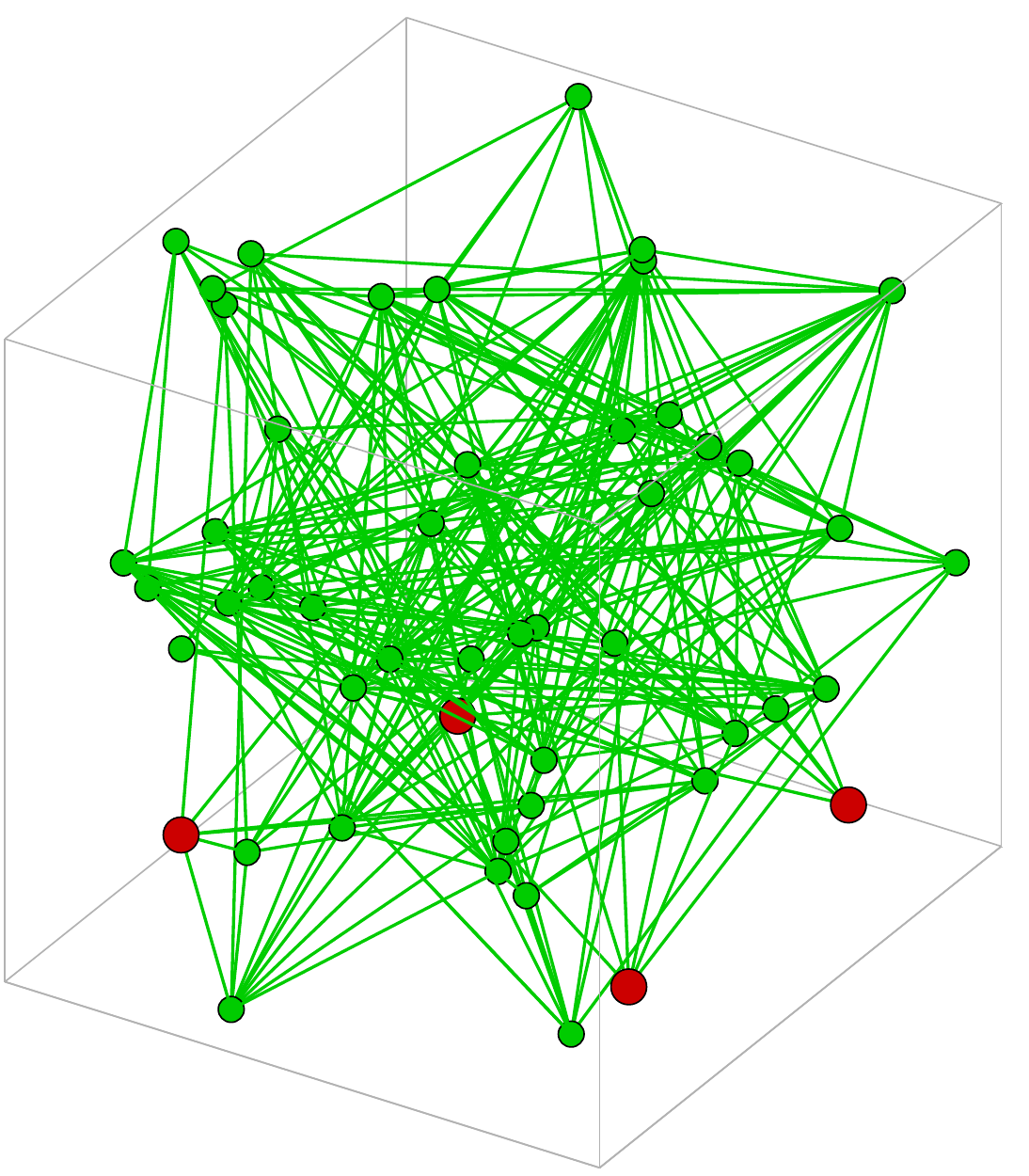}}
  \subfloat[Final localization]{\includegraphics[width=0.25\linewidth]{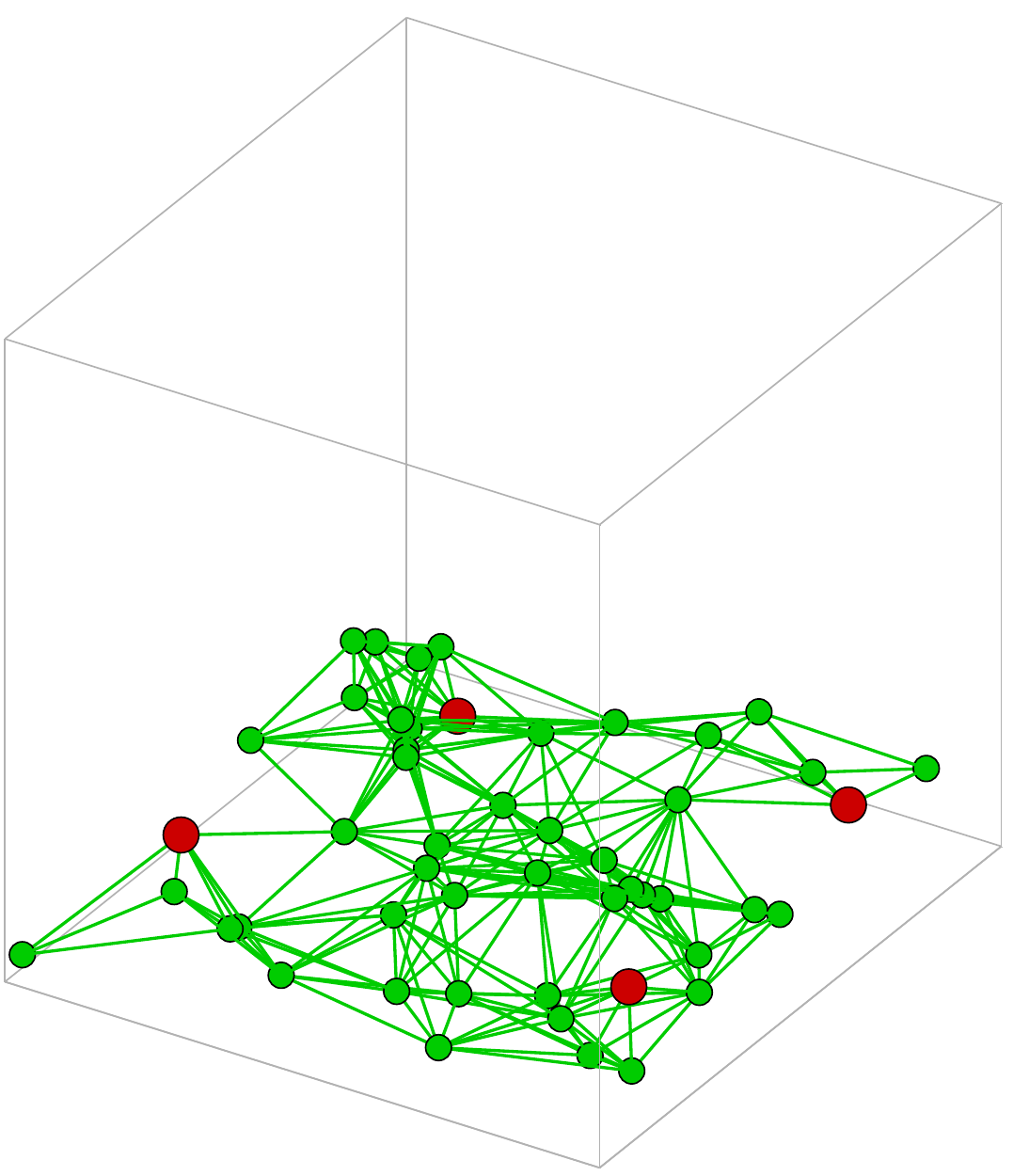}}
  \subfloat[Localization error]{\includegraphics[width=0.25\linewidth]{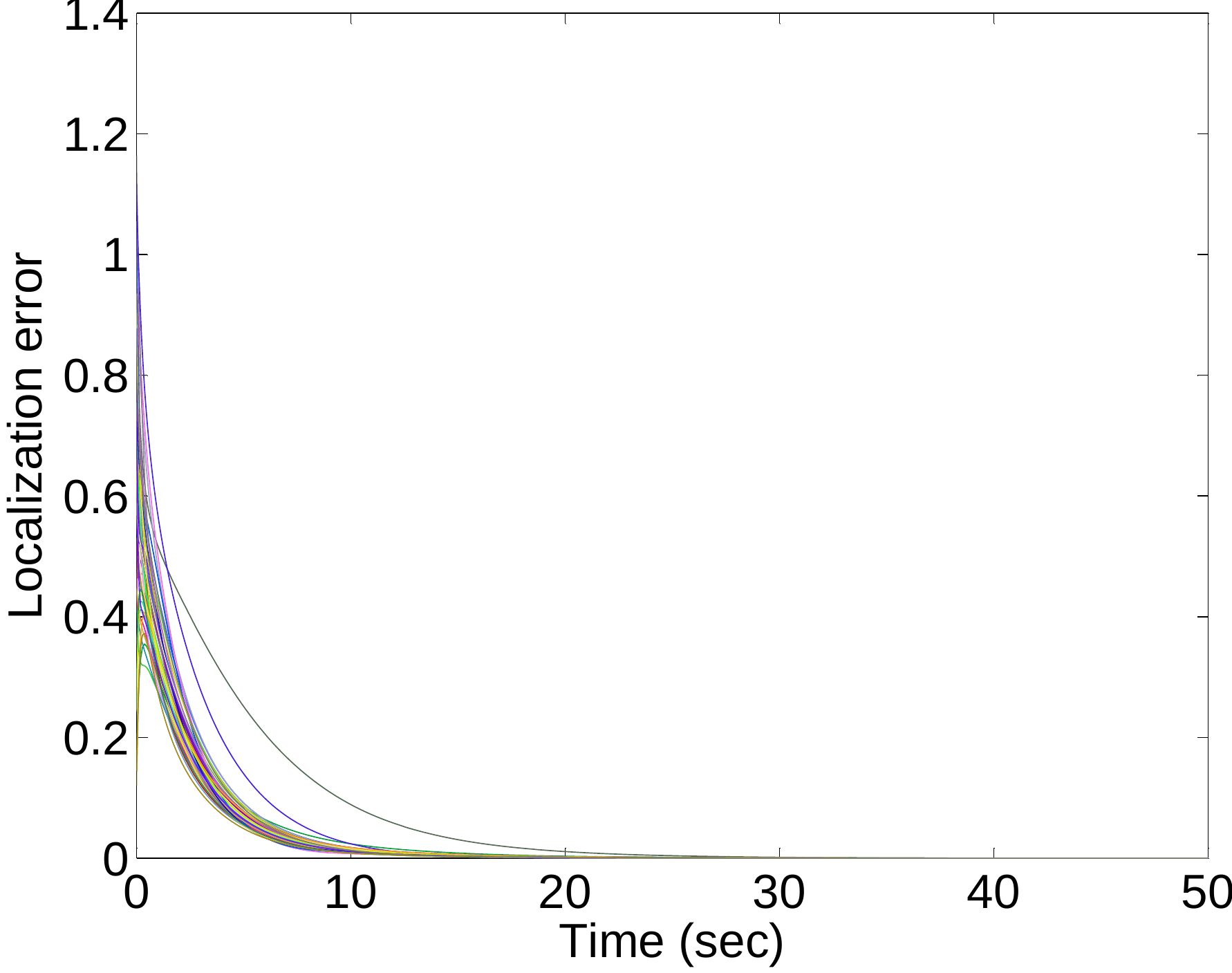}}
  \caption{An example for bearing-based network localization. Green dots: estimates; red dots: anchors.}
  \label{fig_sim_NetworkEst}
\end{figure}

{\small
\section*{Acknowledgements}
The work presented here has been supported by the Israel Science Foundation.

\bibliography{myOwnPub,zsyReferenceAll} 
\bibliographystyle{ieeetr}
}
\end{document}